\documentclass[a4paper,12pt]{amsart}
\usepackage{lmodern}
\usepackage[T1]{fontenc}
\usepackage[english]{babel}
\usepackage{amsmath}
\usepackage{amsthm}
\usepackage{amsfonts}
\usepackage{amssymb}
\usepackage{amsaddr}
\usepackage{microtype} 
\usepackage{MnSymbol}
\usepackage{graphicx}
\usepackage{pict2e}
\usepackage{color}
\usepackage{array}
\usepackage{nth}
\usepackage{transparent}
\usepackage[left=3cm,right=3cm,top=2.8cm,bottom=2.5cm]{geometry}
\usepackage{setspace}
\usepackage{nicefrac}
\usepackage{commath}
\usepackage{natbib}
\usepackage[pdftex]{hyperref}
\usepackage{stmaryrd}
\usepackage{enumerate}
\usepackage{enumitem}
\usepackage[stable]{footmisc}

\hypersetup{colorlinks=true,citecolor=blue,urlcolor=blue,linkcolor=black}
\theoremstyle{definition}
\newtheorem{exmp}{Example}

\theoremstyle{plain}
\newtheorem*{thm*}{See Theorem 1 in \cite{mas1992equilibrium}}
\newtheorem{theorem}{Theorem}
\newtheorem{proposition}{Proposition}

\newtheorem{corollary}{Corollary}

\newtheorem{fact}{Fact}
\newtheorem{claim}{Claim}

\renewcommand{\bfdefault}{b}

\DeclareSymbolFont{Symbols}{OMS}{cmsy}{m}{n}
\DeclareMathSymbol{\Emptyset}{\mathord}{Symbols}{"3B}

\author{\footnotesize{Christian  Basteck$^\dagger$}} 

\thanks{$^\dagger$WZB Berlin, Berlin, Germany. \href{mailto: christian.basteck@wzb.eu}{\tt  christian.basteck@wzb.eu}\\I thank Nick Arnosti, Anna Bogomolnaia, Sean Horan,  Bettina Klaus, Hervé Moulin, Szilvia Pápai, Ludvig Sinander, Alex Teytelboym, William Thomson, Alexander Westkamp and, in particular, Lars Ehlers for inspiring discussions and helpful comments. Financial support by the Deutsche Forschungsgemeinschaft through CRC TRR 190 (project number 280092119) is gratefully acknowledged.}

\title{An Axiomatization of the Random Priority Rule}

\begin{document}

\begin{abstract}
We study the problem of assigning indivisible objects to agents where each is to receive at most one. To ensure fairness in the absence of monetary compensation, we consider random assignments. Random Priority, also known as Random Serial Dictatorship, is characterized by equal-treatment-of-equals, ex-post efficiency and probabilistic (Maskin) monotonicity -- whenever preferences change so that a given deterministic assignment is ranked weakly higher by all agents, the probability of that assignment arising should be weakly larger. Probabilistic monotonicity implies strategy-proofness (in a stochastic dominance sense) for random assignment problems  and is equivalent to it on the universal domain of strict preferences; for deterministic rules it coincides with Maskin monotonicity. 
  
\end{abstract}

\maketitle

\renewcommand{\bfdefault}{b}

\section{Introduction}
Many allocation problems require us to assign indivisible objects to agents such that each agent receives at most one object -- public housing associations assign apartments to tenants, education administrations match teachers to schools, and municipalities assign daycare spots to children.
In the absence of compensating transfers, and in light of the indivisible nature of objects, a desire to treat agents equally leads us to consider randomized assignments.

A random assignment rule prescribes a lottery over deterministic assignments for any possible profile of agents' (strict) preferences over objects. 
Arguably one of the simplest such rules is known as the Random Priority rule ($R\!P$, also known as random serial dictatorship) implemented for example by the following extensive form mechanism: order agents uniformly at random and let each agent, one after another according to the realized ordering, choose their most-preferred among all objects still available. 

It is readily seen to be fair, efficient, and incentive compatible: it treats agents with identical preferences equally from an ex ante point of view,  all assignments arising with positive probability are Pareto-efficient, and choosing the most-preferred available object according to one's true preferences is a dominant strategy -- thus, if the rule is implemented as a direct mechanism, where agents are asked to report their preferences before the mechanism chooses optimally on their behalf, it is strategy-proof.

This paper identifies another normatively appealing, property:  the Random Priority rule is responsive to agents preferences, in that an assignment is more likely to be chosen once agents consider it more preferred. More precisely, our new axiom, probabilistic (Maskin) monotonicity, considers different preference profiles $R, R'$ and an assignment $\mu$ where for each agent   an assignments preferred to $\mu$ under $R'$ is also preferred to $\mu$ under $R$ -- hence, compared to $R$, the assignment $\mu$ has moved up in each agents' ranking of assignments. In any such situation, probabilistic monotonicity requires that the assignment $\mu$ arises with weakly larger probability under $R'$ than under $R$.

As our main result, we find that the random priority rule is the unique random assignment rule satisfying equal-treatment-of-equals, (ex post) efficiency, and probabilistic monotonicity. Moreover, this characterization holds both in an ordinal framework, where agents' preferences are represented by rank order list of objects, as well as in a cardinal framework where preferences are described by von Neumann–Morgenstern utility functions.

Note that this characterisation does not invoke strategy-proofness. This is because probabilistic monotonicity can be seen as a natural strengthening of strategy-proofness, as we argue in our second proposition. In fact, on the universal domain of strict preferences,\footnote{In contrast, our setting assumes that agents are indifferent between different assignments as long as they themselves receive the same object.} probabilistic monotonicity and strategy-proofness are equivalent
-- hence our characterisation of random priority (also known as random serial dictatorship) may be seen as the counterpart to the classic characterisation of random dictatorship as the only random social choice rule on such domains satisfying (ex post) efficiency and strategy-proofness due to \cite{gibbard1977manipulation}. 

Since probabilistic monotonicity is stronger that strategy-proofness on the domain of assignment problems, one may ask whether it can be replaced by strategy-proofness in conjunction with an additional property. This is indeed the case: replacing probabilistic monotonicity with strategy-proofness and a weak, object-wise non-bossiness notion also yields an alternative characterization of $R\!P$ (Theorem \ref{th.ax}).

Many authors have recognized the exceptional position that $R\!P$ occupies among random assignment rules. \cite{abdulkadirouglu1998random} and \cite{knuth1996exact} independently showed that it is not only implemented by a random serial dictatorship as described above but equivalently as the core from random endowments.\footnote{For this, they assume an equal number of agents and distinct objects.} Following up on this surprising result, several authors have shown how randomization over certain deterministic rules, so as to symmetrize the treatment of agents, yields the random priority rule \citep{pathak2011lotteries,lee2011equivalence,carroll2014general,bade2020random}:\footnote{All these assume a unit supply of objects, i.e., do not consider the possibility of copies of objects between which agents are indifferent.} 
in the most general formulation, due to \cite{bade2020random}, symmetrizing any strategy-proof, non-bossy, and efficient deterministic rule yields $R\!P$.

Interestingly, for deterministic rules, strategy-proofness and non-bossiness are equivalent to group-strategy-proofness \citep{pycia2023ordinal} and thus equivalent to Maskin monotonicity \citep{takamiya2007domains}.\footnote{\cite{pycia2023ordinal} and \cite{takamiya2007domains} consider domains that are general enough to allow for the possibility of multiple object copies; the same is the case for our characterization. For a setting with unit supplies, the first equivalence was shown by \cite{papai2000strategyproof} (allowing for an unequal number of objects and agents) and the second by \cite{takamiya2001coalition} (assuming an equal number of objects and agents, i.e., a housing market as in \cite{shapley1974cores}).}

Hence, as randomization over efficient, strategy-proof, and non-bossy (and hence Maskin monotonic) deterministic rules preserves efficiency, in the form of ex post efficiency, and monotonicity, in the form of probabilistic monotonicity, our characterization implies and unifies the previous equivalence results from \cite{abdulkadirouglu1998random} to \cite{bade2020random} (Corollary \ref{cor.eq}).

Further, \cite{bade2016fairness} shows any ex post efficient and symmetric random assignment rule to necessarily violate group-strategy-proofness\footnote{\cite{zhang2019efficient} strengthens this implication by showing that such rules are strongly group manipulable under some mild additional fairness axioms.} so that randomizing over group strategy proof (and ex post efficient) rules to arrive at a symmetrized random assignment rule entails a loss of group strategy-proofness. In contrast, Maskin monotonicity, equivalent to group strategy proofness for deterministic rules, naturally generalizes to probabilistic monotonicity and is in that form preserved under randomization.

\cite{erdil2014strategy} was the first to show that $R\!P$ is not characterized by equal-treatment-of-equals, ex-post efficiency and strategy-proofness in the presence of outside options -- a long-standing conjecture since \cite{bogomolnaiaMoulin2001} proved this to be the case for 3 agents and 3 acceptable objects. \cite{basteck2024constrained} show $R\!P$ not to be characterized by these properties even in the absence of outside options. Hence, to characterize $R\!P$ in conjunction with equal-treatment-of-equals and ex post efficiency, we are forced to strengthen strategy-proofness, for example to probabilistic monotonicity. 

\cite{pycia2024random} consider another strengthening of strategy-proofness -- they show that an obviously strategy-proof mechanism  is symmetric
and (ex post) efficient if and only if it implements $R\!P$. In other word, their characterization strengthens strategy-proofness to obvious strategy proofness (OSP)\citep{li2017obviously}. OSP and probabilistic monotonicity are logically independent. Note that the characterization or $R\!P$ provided by \cite{pycia2024random} refers to the \emph{possibility} of implementing it in an obviously strategy-proof way -- if instead agents are asked to submit their preferences in advance before the outcome is then computed by means of the random priority rule, obvious strategy proofness will no longer be satisfied.\footnote{For example, this procedure is used to assign students to secondary schools in the city of Amsterdam: \href{https://schoolkeuze020.nl/plaatsing/}{https://schoolkeuze020.nl/plaatsing/}} Nonetheless, the latter procedure still ensures probabilistic monotonicity.

Finally, (weak) Maskin monotonicity allows to characterize two of the most prominent families of deterministic assignment rules. \cite{kojima2010axioms} characterize deferred acceptance rules (DA) by non-wastefulness, weak Maskin monotonicity,\footnote{An assignment rule satisfies weak Maskin monotonicity if an assignment that is chosen at some preference profile $R$ and ranked weakly higher by all agents at another profile $R'$ will either still be chosen or replaced by another assignment that Pareto-dominates it, i.e., ranked even higher by all agents.} and population monotonicity. 
Moreover, if weak Maskin monotonicity is strengthened to Maskin monotonicity, the only assignment rules satisfying the three properties are efficient deferred acceptance rules.\footnote{I.e., DA with  Ergin-acyclic priorities \citep{ergin2002efficient} or that satisfies bounded invariance \cite{basteck2024constrained}.} \cite{morrill2013alternative} characterizes top-trading-cycles rules (TTC) by efficiency, weak Maskin
monotonicity, independence of irrelevant rankings, and mutual best. 
Since TTC-rules satisfy Maskin monotonicity \citep{papai2000strategyproof}, one can replace weak Maskin monotonicity by monotonicity in the characterizations. Note that while DA and TTC satisfy Maskin monotonicity -- and hence probabilistic monotonicity -- neither are OSP-implementable in general.


The paper is organized as follows: Section \ref{technicalities} introduces the model; Section \ref{res} presents the results and discusses probabilistic monotonicity in more detail; Section \ref{conclusion} concludes with some remarks on non-symmetric assignment rules.

\section{Model}\label{technicalities}

\subsection{Agents, objects, preferences}\label{AOR}
Let $N$ be a set of agents and $O$ be a set of distinct objects (or object types). For each object $x\in O$, let $q_x$ denote the number of copies. If one wants to allow for the possibility that some agents are not assigned any proper object, we may do so by introducing a null-object, denoted $\emptyset\in O$,\footnote{One may also interpret $\emptyset$ as representing agent-specific outside options.} 
and set $q_\emptyset=|N|$. Thus, it is without loss of generality to assume $\sum_{x\in O} q_x\geq |N|$.

%

Each agent $i\in N$ 
holds strict preferences over the set of objects $O$, represented by a linear order\footnote{I.e., complete, transitive, and antisymmetric ($xR_iy$ and $yR_ix$ together imply $x=y$).}~$R_i$ whose asymmetric part we denote by $P_i$.\footnote{That is, $xP_iy$ whenever $xR_i y$ and $x\neq y$.}  Let $\mathcal{R}^i$ denote the set of all strict preferences of agent $i$ over $O$ and let $\mathcal{R}=\times_{i\in N}\mathcal{R}^i$ denote the set of all preference profiles $R=(R_1,\ldots,R_n)$.  Refer to $\mathcal{R}$ as the unrestricted domain of strict preferences. In contrast, if $O$ includes the null object $\emptyset$ and if, \emph{in addition}, we restrict agents preferences by requiring that $\emptyset$ is ranked last by each agent, we denote this restricted domain as $\mathcal{\underline{R}}$. 
For a given object $x$ and given $R_i$ we write $U(x,R_i)=\{y\in O| yP_i x\}$ for $i$'s (strict) upper contour set at $x$ and denote the (strict) lower contour set as $L(x,R_i)=\{y\in O| xP_i y\}$. 

For any $R_i\in \mathcal{R}^i$ and any $x,y\in O$ adjacent in $R_i$\footnote{Two objects $x,y\in O$ are adjacent in $R_i$, if $U(x,R_i)\backslash\{x,y\}=U(y,R_i)\backslash\{x,y\}$.} we say that $R_i'$ differs from $R_i$ by a swap of $x$ and $y$ if $R_i'|O\backslash\{x\}=R_i|O\backslash\{x\}$, $R_i'|O\backslash\{y\}=R_i|O\backslash\{y\}$ and $yR_i'x \Leftrightarrow xR_i y$. Let $N_{R_i}\subset \mathcal{R}^i$ denote the set or preferences in the neighbourhood of $R_i$, i.e., the set of all preferences that differ from $R_i$ by a swap of two objects adjacent in $R_i$.

\subsection{Random Assignments}\label{ran.ass}

A (deterministic) assignment is a mapping $\mu:N\rightarrow O$ 
 such that each object $x\in O$ is assigned to at most $q_x$ agents, i.e., $|\{i\in N|\mu_i=x\}|\leq q_x$.\footnote{We will write $\mu_i$ instead of $\mu(i)$ for any $i\in N$.}
Let $\mathcal{M}$ denote the set of all assignments. An assignment $\mu\in\mathcal{M}$ is efficient under $R$ if there exists no $\mu'\in \mathcal{M}$ such that $\mu_i'R_i\mu_i$ for all $i\in N$ and $\mu_j'P_j\mu_j$ for some $j\in N$. 
We write $\mathcal{PO}(R)$ for the set of all assignments efficient under $R$.

Let $\Delta(\mathcal{M})$ denote the set of all probability distributions over $\mathcal{M}$ to which we refer to as random assignments. For a given $p\in \Delta(\mathcal{M})$ and $\mu\in \mathcal{M}$, let $p_\mu$ denote the probability assigned to $\mu$ under $p$ and refer to $supp(p)=\{\mu \in \mathcal{M}|p_\mu >0\}$ as the support of $p$. 
We say that a random assignment $p\in \Delta(\mathcal{M})$ is \emph{ex-post efficient} under $R$ if all assignments in its support are efficient, i.e.,  $supp(p)\subseteq \mathcal{PO}(R)$. 

A random assignment $p$ satisfies \emph{equal-treatment-of-equals} if for any two agents $i,j\in N$ such that $R_i=R_j$ and any two deterministic assignments $\mu,\mu'$ such that $\mu_i=\mu'_j$, $\mu_j=\mu'_i$ and $\mu_k=\mu_k'$ for all $k\in N\backslash\{i,j\}$ we have $p_{\mu}=p_{\mu'}$. In other words, deterministic assignments that merely permute the assignments of agents with identical preferences arise with the same probability.\footnote{Alternatively, if one does not wish to distinguish between random assignments that are welfare equivalent, one may define equal-treatment-of-equals with respect to agents' individual object assignment probabilities rather than the probabilities of choosing complete assignments -- see Appendix \ref{app.axioms.we}.}

A random assignment rule $f$ maps preference profiles to random assignment, i.e., $f: \mathcal{R} \rightarrow \Delta(\mathcal{M})$. For a given $R\in \mathcal{R}$ and $\mu\in \mathcal{M}$, let $f_{\mu}(R)$ denote the probability assigned to $\mu$ under $f(R)$. We say that $f$ is ex-post efficient, if $f(R)$ is ex-post efficient for every $R\in \mathcal{R}$. It satisfies equal-treatment-of-equals, if $f(R)$ satisfies equal-treatment-of-equals for every $R\in \mathcal{R}$. It is deterministic, if $|supp(f(R))|=1$ for every $R\in \mathcal{R}$.

\medskip
Besides these intra-profile axioms we will consider four inter-profile axioms that describe how random assignments should respond to changes in agents' preferences.
The first such axioms demands that an assignment should be chosen with higher probability as it becomes more preferred in the eyes of all agents. For that consider two preference profiles $R, R'\in \mathcal{R}$ and a deterministic assignment $\mu$. We say that $R'$ is a $\mu$-monotonic transformation of $R$ if for agents' upper contour sets we have $U(\mu_i, R'_i)\subseteq  U(\mu_i, R_i)$ for all $i$, i.e., if all agents rank their assignment under $\mu$ weakly higher under $R_i'$ than under $R_i$. We say that a random assignment rule $f$ satisfies \emph{probabilistic (Maskin) monotonicity} if for any $R,R' \in \mathcal{R}$ and $\mu \in \mathcal{M}$ such that $R'$ is a $\mu$-monotonic transformation of $R$ we have $f_{\mu}(R') \geq f_{\mu}(R)$.

Note that if $f$ is deterministic, probabilistic monotonicity reduces to Maskin monotonicity for deterministic, single valued rules.\footnote{A deterministic assignment rule $f$ is said to satisfy Maskin monotonicity if for any two profiles $R,R'\in \mathcal{R}$ and assignment $\mu\in \mathcal{M}$ such that $R'$ is a $\mu$-monotonic transformation of $R$, we have that $f(R)=\mu$ implies $f(R')=\mu$.} 

Probabilistic monotonicity strengthens (stochastic dominance) strategy-proofness (see Proposition \ref{prop.decomp}). To define the latter formally, let $p^{ia}$ denote the probability with which agent $i$ is assigned object $a$ under random assignment $p$, and let $p^i=(p^{ia})_{a\in O}$ denote $i$'s individual random assignment. Moreover, given any two random assignments $p,q\in \Delta(\mathcal{M})$ as well as some agent $i$'s preference $R_i$, we say that $p^i$ stochastically $R_i$-dominates $q^i$ if for all $x\in O$,
\[ \sum_{y\in U(x,R_i)} p^{iy} \geq \sum_{y\in U(x,R_i)} q^{iy}. \]
A random assignment $p$ stochastically $R$-dominates another random assignment $q$ if $p^i$ $R_i$-dominates $q^i$ for all $i\in N$. A random assignment rule $f$ is said to be \emph{(sd)-strategy-proof}  if for all $R\in \mathcal{R}$, all $i\in N$ and all $R_i'\in \mathcal{R}^i$,
$f^i(R)$ stochastically $R_i$-dominates $f^i(R_i',R_{-i})$.\footnote{Sd-strategy-proofness is equivalent to the requirement that for any von Neumann–Morgenstern utility function  compatible with a given ordinal ranking of objects, submitting the true ordinal ranking maximizes an agent's expected utility.}

For deterministic assignment rules, non-bossiness demands that if a change in an agents' preferences does not change their own assigned object, it should not change the overall assignment of objects to (other) agents. Our next axiom generalizes this requirement to random assignment rules:  A random assignment rule $f$ is \emph{object-wise non-bossy} if for all $z\in O$, $R, R'_i\in \mathcal{R}$  we have \[f^{iz}(R)=f^{iz}(R'_i,R_{-i})\implies f_{\mu}(R)=f_{\mu}(R'_i,R_{-i}), \forall \mu\in \{\mu'\in \mathcal{M}|\mu'_i=z\}.\] In words, if the change in one agent's preferences does not affect the probability with which they are assigned a particular object $z$, then, conditional on her being assigned that object, it does not affect the probability distribution over (all agents') assignments. Note that for deterministic assignment rules, object-wise non-bossiness is equivalent to non-bossiness, since it has implications only if there is some $z$ such that $f^{iz}(R)=f^{iz}(R'_i,R_{-i})=1$, in which case it demands that the unique assignment $\mu$ for which  $f_{\mu}(R)=1$, arises with probability $1$ also under $f(R'_i,R_{-i})$. 

Weakening the axiom in that it only needs to hold when the change in an agents' preferences is due to a swap of two adjacently ranked objects, we say that a random assignment rule $f$ is \emph{weakly object-wise non-bossy}\footnote{For deterministic, \emph{strategy-proof} assignment rules, weak object-wise non-bossiness can again be shown to be equivalent to non-bossiness.} if for all $x,y,z\in O$, $R\in \mathcal{R}$, and $R'_i\in N_{R_i}$ such that $xP_iy$ and $yP'_ix$  we have \[f^{iz}(R)=f^{iz}(R'_i,R_{-i})\implies f_{\mu}(R)=f_{\mu}(R'_i,R_{-i}), \forall \mu\in \{\mu'\in \mathcal{M}|\mu'_i=z\}.\]

The last intra-profile axiom was introduced by \cite{gibbard1977manipulation} in a setting where agents have strict preferences over all possible outcomes.\footnote{In contrast, in the assignment problem agents have strict preferences over their own assignment.}  A random assignment rule $f$ is \emph{pairwise responsive}  if for all $x,y\in O$, $R\in\mathcal{R}$, $R'_i\in N_{R_i}$ such that $xP_iy$ and $yP'_ix$, and $\mu\in \mathcal{M}$  we have \[\mu_i\notin \{x,y\}\implies f_\mu(R)=f_\mu(R'_i,R_{-i}).\] In words, swapping two adjacently ranked objects in the preference order of some agent may affect the probability of assignments where that agent is assigned one of the two objects but should not change the probability of any assignment where the agent is assigned some other object.

\medskip

One of the most prominent random assignment rules is the Random Priority rule ($R\!P$), also known as random serial dictatorship. For that let $\rhd$ denote a strict priority order over $N$ and let $\Pi$ denote the set of all strict priority orders.
Given $\rhd\in \Pi$, let $f^{\rhd}$ denote the deterministic priority (or serial dictatorship) rule where agents are assigned their most-preferred among all available objects in order of their priority.\footnote{For any $R\in \mathcal{R}^N$ and $i_1\rhd i_2 \rhd \cdots \rhd i_n$, $i_1$ receives her  $R_{i_1}$-most-preferred object in $O$ (denoted by $f_{i_1}^{\rhd}(R)$), and for $l=2,\ldots,n$, $i_l$ receives her  $R_{i_l}$-most-preferred object in $O\backslash \{f_{i_1}^{\rhd}(R),\ldots,f_{i_{l-1}}^{\rhd}(R)\}$ (denoted by $f_{i_l}^{\succ}(R)$).}
Then the Random Priority rule is defined by $R\!P(R) =\frac{1}{n!}\sum_{\rhd\in \Pi}f^{\rhd}(R)$ for all $R\in \mathcal{R}$. Note that any $f^\rhd$ satisfies Maskin monotonicity \citep{kojima2010axioms}, and hence, $R\!P$ satisfies probabilistic monotonicity. Analogously, one finds that as any $f^\rhd$ is strategy-proofness, so is $R\!P$.

\section{Results} \label{res}
\subsection{Axiomatisation}\label{sec.ax}

We are now set to state our main result: ex-post efficiency, equal-treatment-of-equals and some (combination of) inter-profile axiom(s) characterize the Random Priority Rule. Recall that $\mathcal{R}$ denotes the unrestricted domain of strict preferences while, in the presence of a null-object,  $\mathcal{\underline{R}}\subsetneq \mathcal{R}$
 is restricted by the assumption that all agents rank the null-object last.\footnote{While we defined a random assignment rule, as well our axioms, with respect to $\mathcal{R}$, all definitions carry over to the domain $\mathcal{\underline{R}}$ by simply replacing $\mathcal{R}$ by $\mathcal{\underline{R}}$.}
\begin{theorem}\label{th.ax}
Consider a random assignment rule $f$ on $\mathcal{R}$ or $\mathcal{\underline{R}}$,  satisfying ex post efficiency and equal-treatment-of-equals. The following statements are equivalent:

The random assignment rule $f$
\begin{enumerate}
\item satisfies probabilistic monotonicity.
\item satisfies both strategy-proofness and weak object-wise non-bossiness.
\item satisfies pairwise responsiveness.
\item is the Random Priority Rule (on $\mathcal{R}$ or $\mathcal{\underline{R}}$, respectively).
\end{enumerate}

\end{theorem}

The characterization based on probabilistic monotonicity might be seen as the most compelling on normative grounds, as the axiom requires a rule to be responsive to changes in agents' preferences: an assignment should arise with weakly higher probability as it becomes more preferred in the eyes of agents. The second characterization may be seen as useful in that it pins down the difference between the Random Priority Rule and other random assignment rules that satisfy equal-treatment-of-equals, ex post efficiency and strategy-proofness. For example \cite{erdil2014strategy} constructs a rule on the domain of random assignment problems with outside options that (weakly) stochastically $R$-dominates the random priority rule at every profile $R$ in the domain. Similarly, \cite{basteck2024constrained} construct a rule for assignment problems without outside options that dominates the Random Priority Rule for some profiles and yields assignments where agents receive a higher ranked object on average -- if one wants to further explore possibilities of welfare improvements among random assignment rules that satisfy equal-treatment-of-equals, ex post efficiency, and strategy-proofness, Theorem \ref{th.ax} delineates the set of rules to consider: those that violate weak object-wise non-bossiness. Last, the characterization based on pairwise responsiveness pinpoints the precise property, implied by both (1) probabilistic monotonicity and (2) strategy-proofness in conjunction with weak object-wise non-bossiness, that is used in the main step of the proof.

The proof proceeds as follows. We first verify that $(4)\implies (1)$, $(1)\implies (2)$, and $(2)\implies(3)$ -- these steps are straightforward and relegated to appendix. The crucial step is to verify that $(3)\implies (4)$ -- which we sketch below for the case of an equal number of agents and objects and with unit supply of each object, i.e., $|N|=|O|=n$ and $q_x=1$ for all $x\in O$. 

First, let us introduce some additional notation for the set of assignments that arise with different probability under $f$ and $R\!P$ -- for any $R\in \mathcal{R}$, \[\mathcal{M}^\neq(R):=\{\mu\in \mathcal{M}| f_\mu(R)\neq R\!P_\mu (R)\}.\]
Second, observe that ex-post efficiency implies $f_\mu(R)=0= R\!P_{\mu}(R)$ whenever $\mu$ is inefficient (given $R$). Hence, whenever $\mu\in\mathcal{M}^\neq(R)$  for some $\mu$ and $R$, we know that $\mu$ is efficient at $R$.  Third, recall that any assignment $\mu$, efficient at $R$, is the outcome of some (not necessarily unique) deterministic priority rule at $R$ where agents choose their most preferred object among those still available according to some deterministic priority order.

So towards a contradiction suppose $f\neq R\!P$, i.e., that there exists some $R$ and $\mu\in \mathcal{M}^{\neq}(R)$. By the above, we know that $\mu$ must be efficient and hence the outcome under some deterministic priority rule. Let $i$ be the agent who chooses last under the associated priority order and let $\mu_i=z\in O$. Since $i$ chose last, we know that $\mu_j P_j z$ for all $j\neq i$. 

\emph{Step A:} Now consider $R'$ where all $j\neq i$ move $z$ to the bottom of their preference order, in a series of pairwise swaps of adjacent objects. Since these do not involve $\mu_j$, pairwise responsiveness implies that $f_\mu(R')=f_\mu(R)\neq R\!P_{\mu}(R)=R\!P_{\mu}(R')$. Hence, $f$ and $R\!P$ differ not only for some profile $R\in \mathcal{R}$ but for some profile in a restricted domain  where most agents agree on the ranking of some object $z$, namely  \[\mathcal{\tilde{R}}=\{R\in \mathcal{R}|\forall j\neq i, x\in O: xR_jz\}.\]

\emph{Step B:} Next, we will show that we may further narrow down this subdomain while preserving $f\neq R\!P$ (for some profile of the subdomain) -- either there exist profiles where $i$ also ranks $z$ last and for which $f$ and $R\!P$ differ, or we can fix agent $i$'s preferences and know that under any assignment arising with different probability under $f$ and $R\!P$ (at some profile of the remaining subdomain), $i$ will be assigned $z$.

\smallskip
Case (i): Suppose there exists $R\in \mathcal{\tilde{R}}$ and $\mu'\in \mathcal{M}^{\neq}(R)$ such that $\mu'_i P_i z$. Then we may move $z$ to the bottom of $i$'s preference order, in a series of pairwise swaps of adjacent objects. Call the new profile $R^1$. Since the swaps did not involve $\mu'_i$, pairwise responsiveness implies that $f_{\mu'}(R^1)=f_{\mu'}(R)\neq R\!P_{\mu'}(R)=R\!P_{\mu'}(R^1)$. Hence, $f$ and $R\!P$ differ for some profile in the restricted domain  where all agents, including $i$, rank object $z$ last, namely  \[\mathcal{R}^1=\{R\in \mathcal{\tilde{R}}|\forall  x\in O: xR_iz\}.\]

\medskip
Case (ii): Suppose there exists \emph{no} $R\in \mathcal{\tilde{R}}$ and $\mu'\in \mathcal{M}^{\neq}(R)$ such that $\mu'_i P_i z$. If instead there exist $R\in \mathcal{\tilde{R}}$ and $\mu'\in \mathcal{M}^{\neq}(R)$ such that $z P_i \mu'_i$,  we may move $z$ to the top of $i$'s preference order, in a series of pairwise swaps of adjacent objects. Call the new profile $R'$. Since these swaps did not involve $\mu'_i$, pairwise responsiveness implies that $f_{\mu'}(R')\neq R\!P_{\mu'}(R')$. Hence, by ex-post efficiency of $f$ and $R\!P$, $\mu'$ is efficient at $R'$. Yet, there is a Pareto-improving trade where $i$ exchanges $\mu'_i$ for $z$ with $j$ for whom $\mu'_j=z$ and who, as any $j\neq i$, ranks $z$ last -- a contradiction. Thus, for any $R\in \mathcal{\tilde{R}}$ and $\mu'\in \mathcal{M}^{\neq}(R)$ we have $\mu'_i=z$. To narrow down further a subdomain for which $f\neq R\!P$, pick any $R^1\in \mathcal{\tilde{R}}$ for which $\mathcal{M}^{\neq}(R^1)$ is nonempty and define
\[\mathcal{R}^1=\{R\in \mathcal{\tilde{R}}| R_i=R^1_i\}.\]

Regardless of whether we arrived at $\mathcal{R}^1$ via case (i) or (ii), we can then consider $R\in \mathcal{R}^1$ and $\mu\in \mathcal{M}^\neq(R)$, note that $\mu$ is efficient and hence the outcome of some priority rule at $R$, and identify the last agent according to the priority order, who chose an object $y\neq z$ -- call that agent $j$. Analogous to Step $A$ above, we may then move $y$ down to a position immediately above $z$ for most agents -- if Case (i)  applied in the above, we may do so for all agents other than $j$, if Case (ii) applied, we may do so for all agents other than $i$ and $j$. 

Next, analogous to Step $B$ above, we may move down $y$ to just above $z$ also for $j$, as in Case (i), or fix the preferences of $j$, knowing that at all remaining preference profiles where $f$ and $R\!P$ differ, $j$ will be assigned $y$, as in Case (ii) -- in both cases we have arrived at a further restricted preference domain $\mathcal{R}^2$. Define $O_2:=\{z,y\}$ as the set of objects we pushed down in the two iterations, $N_2^{ii}$ the set of agents for whom Case (ii) applied (so $N_2^{ii}\subseteq \{i,j\}$) and define $N_2^{i}:=N\backslash N_2^{ii}$ as the complement. Importantly, for all $k\in N_2^{i}$, and $R\in \mathcal{R}^2$ we know that $oR_kyR_kz$ for all $o\in O\backslash O_2$.

Now, consider $\mathcal{R}^{n-1}$ arrived at by iterating the steps above, together with the associated sets of objects $O_{n-1}$ and agents $N_{n-1}^{ii}$, $N_{n-1}^{i}$. Take any $R\in \mathcal{R}^{n-1}$ for which $\mathcal{M}^{\neq}(R)$ is non-empty. Since the total probabilities over all assignments sum to one,  there exist $\mu,\mu'\in \mathcal{M}^{\neq}(R)$ for which $f_\mu(R)>R\!P_\mu(R)$ and $f_{\mu'}(R)<R\!P_{\mu'}(R)$. Now, by construction, $\mu_k=\mu'_k$ for all $k\in N_{n-1}^{ii}$. Hence, $\mu$ and $\mu'$ differ only in the assignment of agents $k\in N_{n-1}^{i}$. However, all these agents rank the same set of $n-1$ objects at the bottom, and do so in the same order, so that $R_k=R_l$ for all $k,l\in N_{n-1}^{i}$. Thus, equal-treatment-of-equals for $R\!P$ implies $f_\mu(R)>R\!P_\mu(R)=R\!P_{\mu'}(R)>f_{\mu'}(R)$, which violates equal-treatment-of-equals for $f$ -- a contradiction that concludes the sketch of proof for the case of  $q_x=1$ for all $x\in O$ (no copies of objects) and $\sum_{x\in O} q_x=|O|=|N|$.  For the detailed proof of the general case, allowing for an unequal number of objects and agents as well as multiple copies per object, see the appendix. There we also explain how the proof may be further generalized to a domain where all agents are assumed to rank the same (null) object at the bottom

\subsection{Unifying and generalizing previous equivalence results}
\cite{abdulkadirouglu1998random} proved that $R\!P$ is equivalent to the core from random endowments.\footnote{\cite{knuth1996exact} independently proved an analogous result.} Thus, in a setting with an equal number of agents and objects, it arise not only from a (uniform) randomization over $|N|!$ many deterministic priority rules (where agents chose the most preferred among remaining objects according to a deterministic priority order), but equivalently from a randomization over $|N|!$ many TTC-rules that permute agents' roles by permuting their initial endowments. Note that any deterministic priority rule \citep{kojima2010axioms} as well as any TTC-rule \citep{papai2000strategyproof}, satisfies Maskin monotonicity and hence probabilistic monotonicity. Randomizing over $|N|!$ many permutations of each rule, yields a random assignment rule that preserves probabilistic monotonicity as well as ex post efficiency. Moreover, by construction, it also satisfies equal-treatment-of-equals. But then, by Theorem \ref{th.ax}, it must coincide with $R\!P$. 

In the same way, \cite{pathak2011lotteries,lee2011equivalence,carroll2014general,bade2020random} consider progressively larger classes of deterministic rules and show that for each of them, permuting agents' roles and randomizing over all permutations of the rule yields $R\!P$. Since in all these results, the initial deterministic rules under consideration are Maskin monotonic, the results follow from Theorem \ref{th.ax} in the same way as described above. 

To state this formally, let $f$ be a deterministic assignment rule on $\mathcal{R}$ or $\mathcal{\underline{R}}$ and define $\varphi$ on the same domain by permuting agents roles in $f$ and randomizing over all these permutations uniformly, i.e., for all preference profiles $R$ and assignments $\mu$ \[\varphi_\mu(R)=\frac{1}{|N|!}\sum_{\pi\in \Pi}f_{\mu\cdot\pi}(R_{\pi(1)},R_{\pi(2)},... ,R_{\pi(|N|)}),\] where  $\Pi$ denotes the set of all permutations of agents in $N$\footnote{I.e., all bijections $\pi: N\rightarrow N$.} and $\mu\cdot\pi$ denotes the accordingly permuted assignment, i.e., $(\mu\cdot\pi)_i=\mu_{\pi(i)}$ for all $i\in N$.

\begin{corollary}\label{cor.eq}
Consider a deterministic assignment rule $f$ on $\mathcal{R}$ or $\mathcal{\underline{R}}$ and the associated random assignment $\varphi$. If $f$ is Maskin-monotonic (or equivalently strategy-proof and non-bossy) and ex post efficient, then $\varphi$ is the Random Priority Rule $R\!P$.
\end{corollary}

Note that Theorem \ref{th.ax}, and hence Corollary \ref{cor.eq}, allow for multiple copies of objects, thus generalizing the existing equivalence results.\footnote{\cite{abdulkadirouglu1998random,pathak2011lotteries,lee2011equivalence,carroll2014general} consider an equal number of agents and objects and no copies. \cite{bade2020random} allows for an unequal number of agents and objects as well as a null-object in abundant supply with the restriction that the null-object is always ranked last -- this corresponds to the domain $\mathcal{\underline{R}}$ under the restriction that $\emptyset$ is the only object with multiple copies.}   Moroever, Theorem \ref{th.ax} improves on these earlier results also in that it dispenses with the restriction to assignment rules that arise from randomization over deterministic rules and instead  considers  random assignment rules directly. For example, \cite{pycia2015decomposing} show that the set of strategy-proof random assignment rules is strictly larger than the set or randomizations over deterministic strategy-proof assignment rules, so that, in general, restricting to randomizations over deterministic rules is not innocuous when trying to identify random assignments satisfying appealing properties.

\subsection{Expected utility and preference intensities}
Instead of assuming that agents' preferences are given by linear orders over the set of objects $O$, we may also assume that they are described by von Neumann–Morgenstern utility functions, i.e., that there exist $u_i: O\rightarrow \mathbb{R}$ for each $i\in N$. Depending on the interpretation attached, we may thus capture preference intensities or complete agents' preferences  over lotteries of objects by assuming that they compare different lotteries by their corresponding expected utility. Restricting attention to strict preferences over objects, we require  that  $u_i(x)\neq u_i(y)$ for all $i\in N$, $x,y\in O$, and $x\neq y$. Let $\mathcal{U}^i$ denote the set of all such utility functions of agent $i$ and let $\mathcal{U}=\times_{i\in N} \mathcal{U}^i$ denote the set of all utility profiles $u=(u_1,...,u_n)$. The associated strict preference relations and preference profiles are denoted $R^u_i$ and $R^u=(R^u_1,...,R^u_n)$.\footnote{I.e. for all $x,y\in O$ we have $xR^u_i y :\Leftrightarrow u_i(x)\geq u_i(y)$. }

A random assignment rule on the domain $\mathcal{U}$ then maps utility profiles to random assignments, i.e., $f: \mathcal{U}\rightarrow \Delta(\mathcal{M})$. It is ex post efficient if $f(u)$ is ex post efficient for all $u\in \mathcal{U}$ and satisfies probabilistic monotonicity, if probabilistic monotonicity is satisfied for the associated strict preferences $R^u$.\footnote{I.e., $f$ satisfies probabilistic monotonicity if for any $u,u'\in \mathcal{U}$ and $\mu\in \mathcal{M}$ such that $R^{u'}$is a $\mu$-monotonic transformation of  $R^u$, we have $f_{\mu}(u')\geq f_{\mu}(u)$.}  

A random assignment $p$ satisfies equal-treatment-of-equals with respect to $u\in \mathcal{U}$ if for any two agents $i,j\in N$ such that $u_i=u_j$ and any two deterministic assignments $\mu,\mu'$ such that $\mu_i=\mu'_j$, $\mu_j=\mu'_i$ and $\mu_k=\mu_k'$ for all $k\in N\backslash\{i,j\}$ we have $p_{\mu}=p_{\mu'}$. In other words, deterministic assignments that merely permute the roles of agents with identical utility functions arise with the same probability. Note that this allows for unequal treatment of agents with identical ordinal preferences over objects, $R_i^u=R_j^u$, and thus weakens equal-treatment-of-equals as defined with respect to $R$. A random assignment rule satisfies equal-treatment-of-equals on $\mathcal{U}$ if $f(u)$  satisfies equal-treatment-of-equals for any $u\in \mathcal{U}$.

Despite the fact that equal-treatment-of-equals with respect to $u$ is weaker than equal-treatment-of-equals with respect to $R^u$, it gives rise to a characterization of the Random Priority rule $R\!P$ on the domain of (profiles of) von Neumann–Morgenstern utility functions in analogy to Theorem \ref{th.ax}.

\begin{proposition}\label{prop.ax}
Let $f$ be a random assignment rule on $\mathcal{U}$. Then $f$ satisfies equal-treatment-of-equals on $\mathcal{U}$, ex post efficiency, and probabilistic monotonicity if and only if $f$ is the Random Priority Rule.
\end{proposition}

\begin{proof} Crucially, observe that probabilistic monotonicity implies ordinality of $f$, i.e., that for any $\mu\in \mathcal{M}$ and any two $u,u'\in \mathcal{U}$ such that $R^u=R^{u'}$, we have $f_\mu(u)=f_\mu(u')$: since $R^u$ is a $\mu$-monotonic transformation of $R^{u'}$, probabilistic monotonicity demands $f_\mu(u)\geq f_\mu(u')$ and, by a symmetric argument, $f_\mu(u)\leq f_\mu(u')$. 

Hence $f$ takes into account only agents' ordinal preferences $R^u$ rather than the richer information encoded in their von Neumann–Morgenstern utility functions $u$. Thus, to satisfy equal-treatment-of-equals with respect to any $u\in \mathcal{U}$, $f$ has to satisfy equal-treatment-of-equals with respect to the associated preference profiles $R^u$. The claim then follows from Theorem \ref{th.ax}.
\end{proof}

\subsection{Probabilistic monotonicity}\label{sec.pmm}
Given the crucial rule of probabilistic monotonicity in the characterizations above, we complement Theorem \ref{th.ax} and Proposition \ref{prop.ax}  with some observations on the nature of the axiom and its relation to strategy-proofness. For that, observe that
strategy-proofness can be decomposed into the following three axioms \citep{mennle2021partial}: A random assignment rule $f$ satisfies
\begin{itemize}
\item \emph{swap monotonicity} iff for all $i\in N$, $R\in \mathcal{R}$ and $R_i'\in N_{R_i}$ such that $xP_iy$ and $yP_i'x$, we have (i) $f^{iy}(R_i',R_{-i})\geq f^{iy}(R)$ and (ii) $f^{iy}(R_i',R_{-i})= f^{iy}(R)\Rightarrow f^i(R_i',R_{-i})=f^i(R)$. 
\item \emph{upper invariance} iff for all $i\in N$, $R\in \mathcal{R}$ and $R_i'\in N_{R_i}$ such that $xP_iy$ and $yP_i'x$, we have $f^{iz}(R_i',R_{-i})=f^{iz}(R)$ for all $z\in U(x,R_i)$. 
\item \emph{lower invariance} iff for all $i\in N$, $R\in \mathcal{R}$ and $R_i'\in N_{R_i}$ such that $xP_iy$ and $yP_i'x$, we have $f^{iz}(R_i',R_{-i})=f^{iz}(R)$ for all $z\in L(y,R_i)$.
\item \emph{strategy proofness} iff it satisfies swap monotonicity, upper-, and lower invariance \cite[Theorem 1]{mennle2021partial}.
\end{itemize}

Swap monotonicity demands that if an agent swaps to adjacent objects in their preference order, then they should receive the object that is moved up in their ranking with weakly higher probability. In other words, the \emph{sum} over the probabilities of all assignments where they receive the object should be weakly larger. More demanding, probabilistic monotonicity demands that the probability of \emph{each} assignment where they receive the object that has moved up should increase -- instead of comparing the sum of probabilities associated with certain assignments, we now compare each summand. Strengthening the other two axioms analogously, we arrive at the following: 
A random assignment rule $f$ satisfies
\begin{itemize}
\item \emph{assignment swap monotonicity} iff for all $i\in N$, $R\in \mathcal{R}$ and $R_i'\in N_{R_i}$ such that $xP_iy$ and $yP_i'x$, we have (i) $f_{\mu}(R_i',R_{-i})\geq f_{\mu}(R)$ for all $\mu$ with $\mu_i=y$, and (ii) $f_{\mu}(R_i',R_{-i})= f_{\mu}(R)$ for all $\mu$ with $\mu_i=y$ implies $f(R_i',R_{-i})=f(R)$.
\item \emph{upper assignment invariance} iff for all $i\in N$, $R\in \mathcal{R}$ and $R_i'\in N_{R_i}$ such that $xP_iy$ and $yP_i'x$, we have $f_{\mu}(R_i',R_{-i})=f_{\mu}(R)$ for all $\mu$ where $\mu_i\in U(x,R_i)$. 
\item \emph{lower assignment invariance} iff for all $i\in N$, $R\in \mathcal{R}$ and $R_i'\in N_{R_i}$ such that $xP_iy$ and $yP_i'x$, we have $f_{\mu}(R_i',R_{-i})=f_{\mu}(R)$ for all $\mu$ where $\mu_i\in L(y,R_i)$.
\end{itemize}

Note that pairwise responsiveness is equivalent to the conjunction of upper and lower assignment invariance. Keeping the two separate, however, allows us to formulate the next proposition in close analogy to \cite[Theorem~1]{mennle2021partial}.

\begin{proposition}\label{prop.decomp} A random assignment rule satisfies probabilistic monotonicity if and only if it satisfies assignment swap monotonicity, upper assignment invariance, and lower assignment invariance. \end{proposition}

\begin{proof} 
First observe that any violation of upper- or lower assignment invariance also constitutes a violation of probabilistic monotonicity. The same holds if assignment swap monotonicity is violated in that for some $i\in N$, $R\in \mathcal{R}$, $R_i'\in N_{R_i}$ such that $xP_iy$ and $yP_i'x$, and  $\mu$ with $\mu_i=y$ we have $f_{\mu}(R_i',R_{-i})< f_{\mu}(R)$. If instead $f_{\mu}(R_i',R_{-i})= f_{\mu}(R)$ for all $\mu$ with $\mu_i=y$, yet $f(R_i',R_{-i})\neq f(R)$, there is some $\mu'$ with $\mu'_i\neq y$ such that $f_{\mu'}(R_i',R_{-i})> f_{\mu'}(R)$ -- again, this constitutes a violation of probabilistic monotonicity. Hence, to satisfy probabilistic monotonicity, a random assignment rule must satisfy all three properties.

For the other direction, consider a violation of probabilistic monotonicity, i.e., two preference profiles $R,R' \in \mathcal{R}$ and an assignment $\mu \in \mathcal{M}$ such that $R'$ is a $\mu$-monotonic transformation of $R$\footnote{I.e., $U(\mu_i, R'_i)\subseteq  U(\mu_i, R_i)$ for all $i\in N$.} for which we have $f_{\mu}(R') < f_{\mu}(R)$.  Note that we can move from $R$ to $R'$ in a sequence of profiles $R=R^1, R^2, R^3..., R^m=R'$ where at each step two consecutive profiles differ only by a pairwise swap of two adjacent objects in one agents' preferences and where the successor profile is a $\mu$-monotonic transformation of its predecessor. Thus, at some point of the sequence, there are $R^k$ and $R^{k+1}$ such that $f_{\mu}(R^{k+1}) < f_{\mu}(R^{k})$. Let $i$ be the agent for whom $R_i^k$ and $R_i^{k+1}$ differ and denote the two objects adjacent in $R_i^k$ that are swapped as we move to $R_i^{k+1}$ as $x$ and $y$; without loss of generality assume that $x R_i^k y$ and $yR_i^{k+1} x$. 
Since $R^{k+1}$ is a $\mu$-monotonic transformation of $R^k$, we have $\mu_i\neq x$.

Now, if $\mu_i\in U(x, R_i^k)$, the fact that $f_{\mu}(R^{k+1}) < f_{\mu}(R^{k})$ implies that $f$ violates upper assignment invariance. Similarly, if $\mu_i\in L(y, R_i^k)$, we find that $f$ violates lower assignment invariance while if $\mu_i=y$, $f$ violates assignment swap monotonicity. 
Hence, any rule that violates probabilistic monotonicity must violate one of the three properties.
\end{proof}

\begin{corollary}\label{cor} Probabilistic monotonicity implies strategy-proofness. The converse does not hold: there exist strategy-proof random assignment rules for which the rule itself, as well as any other random assignment rule that yields the same individual object assignment probabilities, fails to satisfy probabilistic monotonicity\end{corollary}

The first part is an immediate consequence of Proposition \ref{prop.decomp} above and Theorem~1 in  \citep{mennle2021partial}, that decomposes strategy-proofness into (i) swap monotonicity, (ii) upper invariance, and (iii) lower invariance, each of which is weaker than their counterpart referring to assignments. The fact that probabilistic monotonicity is in fact stronger than strategy-proofness (rather than equivalent to) follows from Theorem~\ref{th.ax} above as well as Proposition~3 in \citep{erdil2014strategy} and Theorem 1 in \citep{basteck2024constrained} which show $R\!P$ \emph{not} to be characterized by equal-treatment-of-equals, ex post efficiency, and  strategy-proofness. While the construction by \cite{erdil2014strategy} relies on the possibility that agents may remain unassigned (or, equivalently, assumes that there is a (null-)object with at least $|N|$ copies), \cite{basteck2024constrained} show that $R\!P$ is not characterized by equal-treatment-of-equals ex post efficiency, and strategy-proofness even in the absence of such outside options.

\medskip
If we do not confine ourself to the domain of random assignment problems -- where agents have strict preferences over their own assigned object but are otherwise indifferent regarding the assignment -- but instead consider general social choice problems on the universal domain of strict preferences, probabilistic monotonicity and strategy proofness are equivalent. Formally, let $N$ be the finite set of agents, $O$ be the finite set of outcomes, let $R_i$ denote agent $i$'s strict preferences over all outcomes, and denote a (strict) preference profile as $R=(R_i)_{i\in N}$. Let $\mathcal{R}$ denote the domain of all such preference profiles. A random social choice rule\footnote{\cite{gibbard1977manipulation} refers to these as \emph{decision schemes}.} is a mapping $f:\mathcal{R}\rightarrow \Delta(O)$ where $\Delta(O)$ denotes the set of probability distributions over $O$; it is strategy-proof if for all $R_i, R_i'$ and $R_{-i}$ we find that $f(R_i,R_{-i})$ stochastically $R_i$-dominates $f(R_i',R_{-i})$.\footnote{I.e., if $\sum_{y\in U(x,R_i)} p_{iy} \geq \sum_{y\in U(x,R_i)} q_{iy}$ where $U(x,R_i)=\{y\in O\backslash \{x\}|yR_i x\}$ denotes the strict upper contour set at $x$.} Equivalently, \cite{gibbard1977manipulation} defines strategy-proofness such that, for any underlying von Neumann–Morgenstern utility function compatible with $R_i$, reporting $R_i$ truthfully maximizes an agents expected utility. A random social choice rule satisfies probabilistic monotonicity if for any two $R,R'$ and $o\in O$, such that $R'$ is an $o$-monotonic transformation of $R$ we find that $f(R')$ chooses $o$ with weakly higher probability than $f(R)$.

\begin{fact}
A random social choice rule $f$, defined on the universal domain of strict preferences, is strategy-proof if and only if it satisfies probabilistic monotonicity.
\end{fact}

This follows, upon closer inspection, from Lemma 2 in \cite{gibbard1977manipulation} (where `non-perverseness' corresponds to swap monotonicity and `localizedness'
 can be shown to be equivalent to  lower- and upper- invariance).\footnote{For deterministic, singleton valued social choice rules on the domain of strict preferences, \cite{muller1977equivalence} show strategy-proofness to be equivalent to `strong positive association', i.e., Maskin monotonicity.} Hence, in the classic characterization of Random Dictatorship as the only random social choice rule on the domain of strict preferences satisfying ex post efficiency and strategy-proofness,\footnote{\cite{gibbard1977manipulation} reports this result as Corollary 1 and credits Hugo Sonnenschein. It is straightforward to see that once we assume anonymity, i.e., require agents to be treated symmetrically, the randomization needs to be uniform.} we may replace the last axiom by the equivalent requirement of probabilistic monotonicity. 
Interestingly, as we leave the universal domain of strict preferences and move to the domain of random assignment problems where agents are indifferent between assignments as long as they award them the same object, probabilistic monotonicity and strategy-proofness diverge -- and it turns out that only probabilistic monotonicity, rather than strategy-proofness, yields a characterization of the Random Priority Rule (also known as Random \emph{Serial} Dictatorship) in analogy to Gibbard's and Sonnenschein's characterization of Random Dictatorship.

\medskip
Finally, let us compare probabilistic monotonicity to obvious-strategy-proofness (OSP) \citep{li2017obviously}. Both strengthen strategy-proofness, with the caveat that OSP applies not to random assignment rules per se, but to possible extensive form mechanisms implementing them.\footnote{More precisely, a random assignment rule is OSP-implementable, if there exists an extensive form mechanism that implements the rule in obviously dominant strategies. Since any obviously dominant strategy is, a fortiori, a dominant strategy, any OSP-implementable rule is strategy-proof.}  We find that both are logically independent: there exist OSP-implementable random assignment rules that violate probabilistic monotonicity and vice versa.

\begin{exmp}
An OSP-implementable 
(random) assignment rule 
violating probabilistic monotonicity. Consider $N=\{1,2,3\}$, $O=\{a,b,c\}$ and a rule $f$ that (i) awards $1$ their most preferred object among $O$ according to $R_1$, (ii) if $1$'s second and third
 most preferred objects 
 are ranked alphabetically, i.e., $aR_1b$, $aR_1c$, or $bR_1c$, then $2$ receives their 
most preferred among the 
remaining objects\footnote{The example does not rely on randomization, but instead constructs a deterministic rule. It can be readily generalized to non-degenerate random assignment rules, e.g., by making the probability with which $2$ gets to choose dependent on $R_1$.}  
while otherwise (iii) $3$ receives their most preferred of the 
remaining objects before (iv) the last 
remaining agent is assigned the last remaining object. This rule can be readily implemented via sequential barter 
\citep{bade2016gibbard} and is hence OSP-implementable (as well as efficient). To see that it violates probabilistic monotonicity, consider the preference profile $R$ 
such that 
$aR_ibR_ic$ 
for all $i\in N$. 
Then $f$ returns 
the assignment 
$\mu=(\mu_1,\mu_2,\mu_3)=(a,b,c)$ 
(with probability $1$). 
If instead we have $R'$ such that $aR'_1cR'_1b$, while as 
before $a R'_i b R'_i c$ for all $i\in \{2,3\}$, then $f$ returns $\mu'=(a,c,b)$ -- despite the fact that $R'$ constitutes a $\mu$-monotonic transformation of $R$.
\end{exmp}

For the other direction, \cite{li2017obviously} shows that the core in Shapley-Scarf housing markets \citep{shapley1974cores}, while implementable in dominant strategies by a  top-trading-cycle mechanism, is not OSP-implementable. Nonetheless, it is (Maskin) monotonic \citep{sonmez1996implementation,takamiya2001coalition}. Similarly, deferred acceptance satisfies monotonicity \citep{kojima2010axioms} but fails to be OSP-implementable \citep{ashlagi2018stable}.

\section{Concluding remarks}\label{conclusion}

Besides being easily implementable and commonly used, the Random Priority Rule ($R\!P$) is the only random assignment rule satisfying probabilistic monotonicity, equal-treatment-of-equals, and ex post efficiency (Theorem \ref{th.ax}). Probabilistic monotonicity has not been considered in the literature so far, but constitutes a natural strengthening of strategy-proofness (Proposition \ref{prop.decomp}) on the domain of assignment problems. 

The fact that our characterisation of $R\!P$ relies on strengthening strategy proofness may be considered particularly compelling given that strategy proofness, together with equal-treatment-of-equals and ex post efficiency, is insufficient to characterise $R\!P$ \citep{erdil2014strategy,basteck2024constrained}\footnote{The construction by \cite{erdil2014strategy} relies on the existence of outside options; \cite{basteck2024constrained} provide a construction without outside options.} -- there are other random assignment rules on the domain of ordinal preferences that satisfy strategy-proofness, equal-treatment-of-equals and ex post efficiency but are not welfare equivalent to $R\!P$. By Theorem \ref{th.ax} these rules must then violate weak object-wise non-bossiness.

The characterization of the Random Priority Rule given in Theorem \ref{th.ax} is the first known characterization for settings with multiple copies of objects, such as for example multiple seats at different schools.\footnote{\cite{pycia2024random} assume an equal number of objects and agents with all objects in unit supply.} 
Moreover, as spelled out in Corollary \ref{cor.eq}, it implies and generalizes several equivalence results according to which randomizing over certain appealing deterministic assignment rules yields  $R\!P$ \citep{abdulkadirouglu1998random,pathak2011lotteries,lee2011equivalence,carroll2014general,bade2020random}. Instead, Theorem \ref{th.ax} characterizes $R\!P$ directly among all random assignment rules, rather than among those that are constructed as randomizations over deterministic rules -- an important distinction since the set of random assignment rules satisfying a particular property may be strictly larger than the set of rules that arises from randomizations over deterministic rules, each of which satisfies the property in question.\footnote{For example, \cite{pycia2015decomposing} show that the set of strategy-proof random assignment rules is strictly larger than the set or randomizations over deterministic strategy-proof assignment rules.}

\appendix
\section{}
\subsection{Proof of Theorem \ref{th.ax}}\label{app.proof}

\begin{proof}{}
\indent $(4)\implies (1)$: We first establish that the random priority rule $R\!P$ satisfies probabilistic monotonicity: this follows from the fact that $R\!P$ is a convex combination of the $n!$ serial dictatorships $f^\rhd$ associated with all possible strict priority orders  over agents, $\rhd\in \Pi$. Each such deterministic rule $f^\rhd$ constitutes a Pareto efficient deferred acceptance rule (with objects' priorities over agents given by $\rhd$) and hence satisfies Maskin monotonicity \citep{kojima2010axioms}. As probabilistic (Maskin) monotonicity is preserved for convex combinations, $R\!P$ satisfies probabilistic monotonicity.

\smallskip
$(1)\implies (2)$: Similarly, strategy-proofness of $R\!P$ follows from the fact that each $f^\rhd$ is strategy-proof. For weak object-wise non-bossiness, consider any swap of two objects in some agents preference order, i.e., $R\in \mathcal{R}$ (resp. $ \mathcal{\underline{R}}$) and $R'_i\in N_{R_i}$ (resp. $ N_{R_i}\cap\mathcal{\underline{R}}$) such that $xP_iy$ and $yP'_ix$. Now, suppose that $f^{iz}(R)=f^{iz}(R'_i,R_{-i})$ for some arbitrary $z\in O$. We have to show that $f_{\mu}(R)=f_{\mu}(R'_i,R_{-i})$ for all $\mu\in \{\mu'\in \mathcal{M}|\mu'_i=z\}$.

If $z=y$, then $(R_i',R_{-i})$ is a $\mu$-monotonic transformation of $R$, so we have $f_{\mu}(R)\leq f_{\mu}(R'_i,R_{-i})$ for each $\mu$ with $\mu_i=z$. Moreover, if for some $\mu'$ with $\mu'_i=z$ we had $f_{\mu'}(R)<f_{\mu'}(R'_i,R_{-i})$, then $\sum_{\mu: \mu_i=z}f_{\mu}(R)<\sum_{\mu: \mu_i=z}f_{\mu}(R'_i,R_{-i})$ contradicting $f^{iz}(R)=f^{iz}(R'_i,R_{-i})$. Analogously if $z=x$.

If instead $z\neq x,y$, then any $\mu$ with $\mu_i=z$, $(R_i',R_{-i})$ is a $\mu$-monotonic transformation of $R$ and vice versa. Hence $f_{\mu}(R)\leq f_{\mu}(R'_i,R_{-i})\leq f_{\mu}(R)$ by probabilistic monotonicity.

\smallskip
$(2)\implies (3)$: For pairwise responsiveness, consider any swap of two objects in some agents preference order, i.e., $R\in \mathcal{R}$ (resp. $ \mathcal{\underline{R}}$) and $R'_i\in N_{R_i}$ (resp. $ N_{R_i}\cap\mathcal{\underline{R}}$) such that $xP_iy$ and $yP'_ix$. By strategy-proofness, this does not affect $i$'s assignment probability for any other object $z\neq x,y$. 
By weak object-wise non-bossiness we have $f_\mu(R)=f_\mu(R'_i,R_{-i})$ for all $\mu\in \mathcal{M}$ s.t. $\mu_i=z$. Thus, $f$ satisfies pairwise responsiveness.

\smallskip

$(3)\implies (4)$: First, recall the notation $\mathcal{M}^\neq(R):=\{\mu\in \mathcal{M}| f_\mu(R)\neq R\!P_\mu (R)\}$. Second, by ex post efficiency of $f$ and $R\!P$, $\mu\in\mathcal{M}^\neq(R)$  implies that $\mu$ is efficient at $R$.  Third, recall that any assignment $\mu$, efficient at $R$, is the outcome of some deterministic priority rule at $R$ where agents choose their most preferred object among those still available according to some deterministic priority order.

Now, towards a contradiction, suppose $f\neq R\!P$, i.e., that there exists some $R$ and $\mu\in \mathcal{M}^{\neq}(R)$. We argue by induction. Suppose  that for $k\geq 0$ we have a triple $\mathcal{R}^k \subseteq \mathcal{R}$, $O_k\subseteq O$, and $N_k\subseteq N$  such that,

\begin{itemize}
\item there exists some $R\in \mathcal{R}^k$ and $\mu\in \mathcal{M}$ such that $f_\mu(R)\neq R\!P_\mu(R)$,
\item $|O_k|=k$ (so increasing in $k$),
\item all agents in $N_k$ rank all objects in $O_k$ at the bottom and in the same order, i.e., 
\begin{itemize}
\item $\forall j\in N_k, o\nin  O_k, o'\in O_k: oP_jo'$,
\item $\forall j,l\in N_k, o,o'\in O_k: oR_jo' \Leftrightarrow oR_lo'$,
\end{itemize}
\item for all other agents, $i\in N\backslash N_k$, preferences are fixed and they receive the same object from $O_k$ under any assignment for which $f$ and $R\!P$ differ, i.e.,
\begin{itemize}
\item $\forall R,R'\in \mathcal{R}^k: R_i=R'_i$,
\item $\forall R\in  \mathcal{R}^k, \mu,\mu'\in \mathcal{M}^\neq(R): \mu_i=\mu'_i\in O_k$.
\end{itemize}

\end{itemize}

We want to show that there exists an analogous triple for $k'=k+1$ with $O_{k+1}\supset O_k$ and $N_{k+1}\subseteq N_k$ until we find a triple that allows us to derive a contradiction: intuitively, for any $R\in \mathcal{R}^k$ where there is a difference between $f$ and $R\!P$, that difference must be limited to the assignment of objects to agents in $N_k$ -- and since their preferences are increasingly aligned as $O_k$ expands, we eventually arrive at a contradiction to equal-treatment-of-equals.

For the induction base, we consider $k=0$, $\mathcal{R}^k=\mathcal{R}$, $O_k=\{\}$, and $N_k=N$ -- and are relieved to find that there is nothing to show as the existence of some $R\in \mathcal{R}$ and $\mu\in \mathcal{M}$ such that $f_\mu(R)\neq R\!P_\mu(R)$ was the starting assumption of our proof by contradiction.

For the induction step, suppose we have $\mathcal{R}^k$, $O_k$, and $N_k$ as described above with $|O_k|<|O|-1$ (so that for agents in $N_k$, preferences are not yet completely determined) and $|N_k|\geq 1$ (so that there is some agent whose preferences are not yet completely determined). Consider any $R\in \mathcal{R}^k$ and $\mu\in \mathcal{M}^{\neq}(R)$. Since $\mu\in \mathcal{M}^{\neq}(R)$, agents in $N\backslash N_k$ are assigned objects in~$O_k$. Moreover $\mu$ is efficient and hence the outcome under some deterministic priority rule at $R$. Since agents in $N_k$ rank objects in $O_k$ last, they will choose objects in $O_k$ only after all objects in $O\backslash O_k$ are assigned. Let $i\in N_k$ be the last agent to choose an object in $O\backslash O_k$ under the associated priority order and denote the object as $y=\mu_i$. Then, for all $j\in N_k$ who choose before $i$, we have $\mu_j R_j y P_j o$  for all $o\in O_k$ while for all $j\in N_k$ who choose after $i$ we have $y P_j \mu_j$, $\mu_j\in O_k$. Set $N_k^y=\{j'\in N_k|\mu_j=y\}$. 

Now let all $j\in N_k\backslash N_k^y$ move $y$ down in their preference order to a position immediately above objects in $O_k$, and do so in a series of pairwise swaps of adjacent objects. Call this new profile $R'$. Since swaps did not involve $\mu_j$, pairwise responsiveness implies $f_\mu(R')=f_\mu(R)\neq R\!P_{\mu}(R)=R\!P_{\mu}(R')$. Hence, $f$ and $R\!P$ differ for some profile in a further restricted domain where weakly more than $|N_k|-|q_y|$ agents in $N_k$ agree on the ranking of objects $O_k\cup\{y\}$ at the bottom, namely  \begin{align*}
\mathcal{\tilde{R}}^k=\{R\in \mathcal{R}^k|\forall j\in N_k\backslash N_k^y, o\nin  O_k\cup\{y\}, o'\in O_k\cup\{y\}: oP_jo'\\
 \quad \text{and}  \quad
\forall j,l\in N_k\backslash N_k^y, o,o'\in O_k\cup\{y\}: oR_jo' \Leftrightarrow oR_lo'\}.
\end{align*}

\noindent To further narrow down the subdomain, denote agents in $N_k^y$ as $N_k^y=\{j_1,\cdots,j_{|N_k^y|}\}$ and consider $j_1$. There are two cases:

\medskip

Case (i): Suppose there exists $R\in \mathcal{\tilde{R}}^k$ and $\mu'\in \mathcal{M}^{\neq}(R)$ such that $\mu'_{j_1} P_{j_1} y$ or $\mu'_{j_1}\in O_k$. Then we may move $y$ down in $j_1$'s preference order, to a position immediately above objects in $O_k$, in a series of pairwise swaps of adjacent objects. Call the new profile $R^{k.1}$. Since the swaps did not involve $\mu'_{j_1}$, pairwise responsiveness implies that $f_\mu'(R^{k.1})\neq R\!P_{\mu'}(R^{k.1})$. Hence, $f$ and $R\!P$ differ for some profile on 
\begin{align*}
\mathcal{\tilde{R}}^{k.1}=\{R\in \mathcal{\tilde{R}}^k\subseteq \mathcal{R}^k|\forall  o\nin  O_k\cup\{y\}, o'\in O_k\cup\{y\}: oP_{j_1}o'\\
 \quad \text{and}  \quad
\forall j\in N_k\backslash N_k^y, o,o'\in O_k\cup\{y\}: oR_{j}o' \Leftrightarrow oR_{j_1}o'\}.
\end{align*}

\medskip

Case (ii): Suppose instead that for any $R\in \mathcal{\tilde{R}}^k$, $\mu' \in \mathcal{M}^{\neq}(R)$ 
implies $yR_{j_1}\mu'_i$ and $\mu'_{j_1}\nin O_k$. If $y P_{j_1} \mu'_i$,  
we may move $y$ to the top of ${j_1}$'s preference order, in a series of pairwise swaps. Call the new profile $R'$. 
Since the swaps did not involve $\mu'_{j_1}$, we have
$f_{\mu'}(R')\neq R\!P_{\mu'}(R')$. Hence, by ex-post efficiency of $f$ and $R\!P$,
$\mu'$ is efficient at $R'$. Hence, under $\mu'$,  all copies of $y$ must be assigned to agents other than $j$. Since $y\nin O_k$, it must be assigned to agents in $N_k$ and since $|N_k^y\backslash\{j_1\}|<q_y$ there must be some agent $j\in N_k\backslash N_k^y$ who is assigned $y$ (and ranks $y$ immediately above all objects in $O_k$, i.e., below all other objects in $O\backslash O_k$) -- but then there is a Pareto-improving trade between $j$ and $j_1$, a contradiction. Thus, for any 
$R\in \mathcal{\tilde{R}}^k$ and $\mu'\in \mathcal{M}^{\neq}(R)$ we have $\mu'_{j_1}=y$. Now, pick any $R^{k.1}\in \mathcal{\tilde{R}}^k$ 
for which $\mathcal{M}^{\neq}(R^{k.1})\neq \emptyset$ and define
\[\mathcal{\tilde{R}}^{k.1}=\{R\in \mathcal{\tilde{R}}^k\subseteq \mathcal{R}^k| R_{j_1}=R^{k.1}_{j_1}\}.\]

Next, take $j_2$. Again one of two cases applies, analogous to the two cases above where the initial subdomain is now $\mathcal{\tilde{R}}^{k.1}$ instead of $\mathcal{\tilde{R}}^k$ and the new subdomain is $\mathcal{\tilde{R}}^{k.2}$. Continue in this way to arrive at $\mathcal{R}^{k+1}:=\mathcal{\tilde{R}}^{k.|N_k^y|}$. To complete the induction step, set $O_{k+1}:=O_k\cup\{y\}$ and remove all agents for which Case (ii) applied from $N_{k}$ to arrive at $N_{k+1}$. 

\medskip

Finally, repeat the induction step until we arrive at $\mathcal{R}^{\bar{k}}$, $O_{\bar{k}}$, and $N_{\bar{k}}$ for which $\bar{k}=|O|-1$  or for which $|N_{\bar{k}}|= 0$.

\medskip

If $\bar{k}=|O|-1$, take any $R\in \mathcal{R}^{\bar{k}}$ for which $\mathcal{M}^{\neq}(R)$ is non-empty and consider $\mu,\mu'\in \mathcal{M}^{\neq}(R)$ for which $f_\mu(R)>R\!P_\mu(R)$ and $f_{\mu'}(R)<R\!P_{\mu'}(R)$. By construction, $\mu_j=\mu'_j$ for all $j\nin N_{\bar{k}}$. Hence, $\mu$ and $\mu'$ differ only in the assignment of agents in $N_{\bar{k}}$. However, all these agents rank the same set of $|O|-1$ objects at the bottom, and do so in the same order, so that $R_i=R_j$ for all $i,j\in N_{\bar{k}}$. Thus, equal-treatment-of-equals for $R\!P$ implies $f_\mu(R)>R\!P_\mu(R)=R\!P_{\mu'}(R)>f_{\mu'}(R)$, which violates equal-treatment-of-equals for $f$ -- a contradiction.

\medskip

So instead we must have $|N_{\bar{k}}|= 0$. In this case preferences of all agents are completely determined, i.e., $\mathcal{R}^{\bar{k}}=\{R\}$ is a singleton. Moreover since all agents are in $N\backslash N_{\bar{k}}$,  we find that for every $\mu\in \mathcal{M}^\neq(R)$, all agents receive the same object -- i.e.,  $\mathcal{M}^{\neq}(R)=\{\mu\}$ is a singleton as well. But if $f_\mu(R)\neq R\!P_\mu(R)$ (while for all other assignments $\mu'$ we have $f_{\mu'}(R)=R\!P_{\mu'}(R)$), this contradicts the fact that under both $f(R)$ and $R\!P(R)$, the total probability of all assignments must sum to 1.

\medskip
This completes the proof for the case of the unrestricted domain of strict preferences $\mathcal{R}$. If we consider random assignment rules on $\underline{\mathcal{R}}$ where all agents rank the same (null) object $\emptyset$ last, the proof proceeds analogous, except that we may choose the induction base  as $k=0$, $\mathcal{R}^k=\mathcal{R}$, $O_k=\{\emptyset\}$, and $N_k=N$, i.e., we start `one step ahead' in our attempt to narrow down the domain of preferences by making preferences as similar as possible for agents in $N_k$ (while preserving $f\neq R\!P$).
\end{proof}

\subsection{Welfare equivalent random assignments}\label{app.axioms.we}
Since agents' have preferences over their own assigned object -- but are indifferent between assignments apart from that -- the literature on random assignments has often focussed on agents' individual object assignment probabilities, rather than on the probabilities with which complete deterministic assignments are chosen.

Recall that $p^{ia}$ denotes the probability with which agent $i$ is assigned object $a$ under a given random assignment $p$  and that $p^i=(p^{ia})_{a\in O}$ denotes $i$'s individual random assignment. Two random assignments $p,q$ are said to be (ex ante) welfare-equivalent if $p^i=q^i$ for all $i\in N$, i.e., if they agree on agents' individual object assignment probabilities. Similarly, two random assignment rules, $f$ and $g$, are welfare equivalent, if $f(R)$ and $g(R)$ are welfare equivalent for every $R$.\footnote{Some papers directly define  a random assignment as a matrix $(p^{ia})_{i\in N, a\in O}$ rather than as a convex combination of deterministic assignments. To the extend that on is interested in analysing random assignments only up to welfare equivalence, this is well defined as any such matrix $(p^{ia})_{i\in N, a\in O}$ can be represented as a convex combination of deterministic assignments by (a slight generalization of) the Birkhoff-von Neumann Theorem \citep{birkhoff1946three}.} 

Welfare-equivalence gives rise to equivalence classes of random assignments that we denote as $[p]:=\{q\in \Delta(\mathcal{M})|\forall i\in N\!: \, q_i=p_i\}$ and that can be identified with the collection of individual object assignment probabilities $(p^{ia})_{i\in N, a\in O}$. A collection of individual object assignment probabilities $[p]$ is then typically said to be ex post efficient given $R$, if there \emph{exists} some random assignment $q\in [p]$ that is ex-post-efficient with respect to $R$, i.e., if there exist a convex combination of deterministic, Pareto-efficient assignments giving rise to the individual object assignment probabilities $(p^{ia})_{i\in N, a\in O}$. 

Moreover, for given $R$, $[p]$ is typically said to satisfy equal-treatment-of-equals if $R_i=R_j$ implies $p^{ia}=p^{ja}$ for all $i,j\in N$, $a\in O$. Equivalently, we can define equal-treatments-of-equals for a collection of individual object assignment probabilities in analogy to ex post efficiency, i.e., say that the property holds for $[p]$ if it holds for some $q^*\in[p]$:

\begin{claim}\label{claim.ete}
Consider a preference profile $R\in \mathcal{R}$ and a collection of individual assignment probabilities $[p]$. The following statements are equivalent:
\begin{itemize}
\item[(i)] $\forall i,j\in N, a\in O: \quad R_i=R_j \implies p^{ia}=p^{ja}$.
\item[(ii)] there exists a random assignment $q^*\in [p]$  satisfying equal-treatment-of-equals.
\end{itemize}
\end{claim}

\begin{proof}[Proof of Claim \ref{claim.ete}]
$(ii)\implies (i)$: Take $q\in[p]$ such that $q$ satisfies-equal-treatement-of-equals and suppose there are $i,j\in N$ for which $R_i=R_j$. Now for any $a$, \[q^{ia}=\sum_{\mu:\mu_i=a}q_\mu=\sum_{b\in O}\left(\sum_{\mu:\mu_i=a, \mu_j=b}q_\mu\right)=\sum_{b\in O}\left(\sum_{\mu:\mu_i=b, \mu_j=a}q_\mu\right)=\sum_{\mu:\mu_j=a}q_\mu=q^{ja}\]
where the central equality follows from the definition of equal-treatment-of-equals of random assignments.

$(i)\implies(ii)$: Suppose that there exist a set of agents $M\subseteq N$ with identical preferences, i.e., $R_i=R_j$ for all $i,j\in M$. Let $q\in \Delta(M)$ be a random assignment, i.e., a probability distribution over deterministic assignments giving rise to individual object assignment probabilities $q^{ia}=p^{ia}$ for all $i\in N, a\in O$. Let $\Pi$ denote the set of all permutations of agents in $M$.\footnote{A permutation $\pi$ on $M$ is a bijection $\pi: M\rightarrow M$.} Now, for each assignment $\mu$ in the support of $q$, and each $\pi\in \Pi$, define $\mu^\pi$ as the assignment where $\mu^\pi_i=\mu_{\pi(i)}$. Last, construct the random assignment $q_M$ from $q$ by replacing each assignment $\mu$ with weight $q_\mu$ by the $|M|!$ many assignments $\mu^{\pi}$, $\pi\in \Pi$, each with weight $\frac{q_\mu}{|M|!}$. Then, by construction, $q_M$ satisfies equal-treatment-of-equals among all agents in $M$. Moreover $q_M^i=q^i$ for all  $i\in M$ as well as all $i\notin M$ (since for them the permutation did not change their assignment, i.e., $\mu_i=\mu^{\pi}_i$ for all $\mu$ and $\pi\in \Pi$). Thus $q_M\in [p]$. Repeating the construction for all sets of agents with identical preferences $M'\subseteq N$, $M'\neq M$, we arrive at a random assignment $q^*\in [p]$ that satisfies equal-treatment-of-equals for all $i,j\in N$.

\end{proof}

Hence, if one wants to distinguish between random assignments only up to welfare-equivalence, defining equal-treatment-of-equals with respect to the probability distribution over deterministic assignments, as we do in Section \ref{ran.ass}, is no more restrictive than defining equal-treatments-of-equals for a collection of individual assignment probabilities -- whenever the latter satisfies equal-treatment-of-equals, we could just as well consider some decomposition into deterministic assignments that satisfies the equal-treatments-of-equals as defined in Section \ref{ran.ass}.

\bibliographystyle{plainnat}
\bibliography{./library}

\end{document}